\documentclass[11pt]{article}

\usepackage{amsmath, amssymb, amsfonts, hyperref, graphicx, verbatim, fullpage,latexsym, amsmath, amsfonts, amssymb, amsthm, graphicx, epsfig, array, longtable, calc, multirow, hhline, authblk}
\usepackage[lined,boxed,linesnumbered]{algorithm2e}

\title{\bf A Randomized Exchange Algorithm for Computing Optimal Approximate Designs of Experiments}
\author[1,2]{Radoslav Harman}
\author[1]{Lenka Filov\'{a}}
\author[3,4,5]{Peter Richt\'{a}rik}
\affil[1]{Comenius University in Bratislava, Slovakia}
\affil[2]{Johannes Kepler University Linz, Austria}
\affil[3]{King Abdullah University of Science and Technology, Kingdom of Saudi Arabia}
\affil[4]{University of Edinburgh, United Kingdom}
\affil[5]{Moscow Institute of Physics and Technology, Russia}

\usepackage{fullpage}

\def\a{\mathbf a}
\def\d{\mathbf d}
\def\e{\mathbf e}
\def\f{\mathbf f}
\def\g{\mathbf g}
\def\k{\mathbf k}
\def\l{\mathbf l}
\def\x{\mathbf x}
\def\y{\mathbf y}
\def\w{\mathbf w}
\def\I{\mathbf I}
\def\M{\mathbf M}
\def\R{\mathbb R}
\def\V{\mathbf V}
\def\X{\mathfrak X}
\def\1{\mathbf 1}
\def\0{\mathbf 0}
\def\tr{\mathrm{tr}}
\def\eff{\mathrm{eff}}
\def\supp{\mathrm{supp}}
\def\argmax{\mathrm{argmax}}
\def\argmin{\mathrm{argmin}}

\newtheorem{theorem}{Theorem}
\newtheorem{proposition}{Proposition}

\newtheorem{lemma}{Lemma}

\def\cL{\mathcal L}
\newcommand{\eqdef}{:=}

\begin{document}
\maketitle

\begin{abstract}
	We propose a class of subspace ascent methods for computing optimal approximate designs that covers both existing as well as new and more efficient algorithms. Within this class of methods, we construct a simple, randomized exchange algorithm (REX). Numerical comparisons suggest that the performance of REX is comparable or superior to the performance of state-of-the-art methods across a broad range of problem structures and sizes. We focus on the most commonly used criterion of D-optimality that also has applications beyond experimental design, such as the construction of the minimum volume ellipsoid containing a given set of data-points. For D-optimality, we prove that the proposed algorithm converges to the optimum. We also provide formulas for the optimal exchange of weights in the case of the criterion of A-optimality. These formulas enable one to use REX for computing A-optimal and I-optimal designs.
\end{abstract}

\textbf{Keywords:} optimal approximate design of experiments, D-optimality, A-optimality, I-optimality, convex optimization, minimum volume ellipsoid

\section{Introduction}

The topic of this paper is the computation of optimal approximate designs for regression models with uncorrelated errors (e.g., \cite{Fedorov}, \cite{Pazman},  \cite{Puk}, \cite{Atkinson}). As a special case, we also briefly discuss the minimum volume enclosing ellipsoid problem (e.g., \cite{Todd-book}).
\bigskip

Suppose that we intend to perform an experiment consisting of a set of trials. Assume that the observed response of each trial depends on a design point $x$ chosen from a finite set $\mathfrak{X}$. For simplicity but without the loss of generality, we will assume that $\X=\{1,\ldots,n\}$. Note that in the experimental design problems, $n$ can be a large number, often many thousands. 
\bigskip

For $x \in \mathfrak{X}$, the real-valued observation $Y(x)$ is assumed to satisfy the linear\footnote{It is also straightforward to apply the algorithms proposed in this paper to the computation of locally-optimal designs for \emph{non-linear} regression models. In this sense, we gain clarity and do not lose generality if we formulate our method for the linear regression (see \cite{Atkinson}, chapter 17).} regression model $Y(x)=\f'(x)\beta +\varepsilon(x)$, where $\f(x) \in \mathbb{R}^m$ is the known regressor associated with $x$, the vector $\beta \in \mathbb{R}^m$ contains the unknown parameters of the model, and $\varepsilon(x) \sim N(0,\sigma^2)$, with $\sigma^2>0$, is an unknown variance.\footnote{We assume homoskedasticity and the normality of errors for the sake of simplicity. Our algorithms can be applied to all models with uncorrelated $\epsilon(x)$, $x \in \X$, with finite variances, after a simple transformation of the problem (\cite{Atkinson}, chapter 23).} For different trials, the errors are assumed to be uncorrelated. Let the model be non-singular in the sense that $\{\f(x): x \in \mathfrak{X}\}$ spans $\mathbb{R}^m$. We also avoid redundant regressors by assuming that $\f(x) \neq \0_m$ for all $x \in \X$. Note that in the applications, the number $m$ of the parameters tends to be relatively small, mostly less than $10$.
\bigskip

We formalize an (approximate experimental) design on $\X$ as an $n$-dimensional vector $\w$ with non-negative real components summing to one.\footnote{Often, an approximate experimental design is formalized as a probability measure on $\X$, but for a finite design space, the representations by a probability measure and by a vector of probabilities are equivalent.} From the point of view of an experimenter, the component $w_x$ of $\w$ represents the proportion of the trials to be performed in the design point $x \in \X$. The support of a design $\w$ is $\supp(\w)=\{x \in \X: w_x>0\}$.\footnote{Note that there are theoretical reasons corroborated by extensive numerical evidence that the ``optimal'' designs $w$ tend to be sparse in the sense of having support much smaller than $n$. More precisely, there always exists an optimal design supported on a set of a size that does not exceed $1+m(m+1)/2$.} The set of all designs on $\mathfrak{X}$ denoted by $\Xi$ is the probability simplex in $\R^n$, which is compact and convex.
\bigskip

The information matrix associated with a design $\w$ is defined by
\begin{equation*}\label{infmat}
\M(\w) \eqdef \sum_{x\in \mathfrak{X}}w_x \f(x)\f'(x).
\end{equation*}
Under our model's assumptions, the information matrix is proportional to the Fisher information matrix corresponding to $\beta$. Therefore, the general aim is to choose $\w$ such that $\M(\w)$ is ``as large as possible'', which we make precise next. 
\bigskip

Let $\mathcal{S}^m_+$ be the set of $m\times m$ symmetric non-negative definite matrices, and let $\Phi: \mathcal{S}^m_+ \to \mathbb{R} \cup \{-\infty\}$ be a criterion of optimality, that is, a function measuring the ``size'' of information matrices. A design $\w^*$ is said to be a $\Phi$-optimal design if it maximizes $\Phi(\M(\w))$ in the class $\Xi$ of all designs:
\begin{equation}\label{eq:OPT} 
\w^* \in \argmax\{\Phi(\M(\w)): \w\in \Xi\}.
\end{equation} 
Matrix $\M(\w^*)$ is referred to as the $\Phi$-optimal information matrix.
\bigskip

Due to their natural statistical interpretations, the two most common optimality criteria are $D$-optimality and $A$-optimality. In this paper, we use them in the forms (e.g., \cite{Puk}) $\Phi_D(\M)=(\det(\M))^{1/m}$ and $\Phi_A(\M)=(\mathrm{tr}(\M^{-1}))^{-1}$, respectively, for any positive definite matrix $\M$. For a singular non-negative definite matrix $\M$, the values of the criteria are defined to be $0$. Functions $\Phi_D$ and $\Phi_A$ are positively homogeneous, continuous, and concave on the set of all non-negative definite matrices. Hence, for both $D$- and $A$-optimality, the optimal design always exists, and the optimal information matrix is non-singular. 
\bigskip

Let $\Xi_R$ denote the set of regular designs, i.e.,
\begin{equation*}
  \Xi_R=\{\w \in \Xi: \M(\w) \text{ is non-singular}\}.
\end{equation*}
For this paragraph, let $\w \in \Xi_R$ be fixed. The variance function of $\w$ is the $n$-dimensional vector $\d(\w)$ with components
\begin{equation*}
  d_x(\w)= \f'(x)\M^{-1}(\w)\f(x), \: x \in \X.
\end{equation*}
Therefore, $\d(\cdot)$ is a function that maps regular designs to $\R^n$. From the statistical point of view, $d_x(\w)$ is proportional to the variance of the best linear unbiased estimator (BLUE) of $\f'(x)\beta$, provided that the approximate design $\w$ can be actually carried out (hence the term). We also note that
\begin{equation}
  d_x(\w)=\lim_{\alpha \to 0_+} \frac{\Psi_D[(1-\alpha)\M(\w)+ \alpha\M(\e_x)]-\Psi_D(\M(\w))}{\alpha} + m,
\end{equation}
where $\Psi_D(\M)=\log\det(\M)$ is a version of the $D$-optimality criterion and $\e_x$ is the singular design in $x$ (formally the $x$-th standardized unit vector, i.e., the vertex of $\Xi$). This explains why $\d$ is often used in iterative algorithms for computing the $D$-optimal designs to select ``the most promising direction'' of change of the current design.   
\bigskip

For $A$-optimality, because the directional derivative of $\Psi_A(\M)=-\tr(\M^{-1})$ in $\M(\w)$ in the direction of $\M(\e_x)$ is $\f'(x)\M^{-2}(\w)\f(x)-\tr(\M^{-1}(\w))$, we can define an $n$-dimensional vector $\a(\w)$ with components
\begin{equation*}
 a_x(\w)=\f'(x)\M^{-2}(\w)\f(x), \: x \in \X
\end{equation*}
as a quantity analogous to the variance function in the case of $D$-optimality. 
\bigskip

In the following, let $\g(\w)$ be either $\d(\w)$ or $\a(\w)$, depending on the criterion under consideration. Besides indicating a degree of importance of design points, another reason for utilizing the vector $\g(\w)$ is that $\max_x g_x(\w)$ can be used to compute a natural stopping rule for algorithms based on the notion of statistical efficiency, as follows.
\bigskip

The $D$- and $A$-efficiencies of a design $\w$ relative to $\bar{\w}\in \Xi_R$ are (see \cite{Puk}) 
\begin{equation*}
\mathrm{eff}_D(\w \:|\: \bar{\w}) \eqdef \frac{\Phi_D(\M(\w))}{\Phi_D(\M(\bar{\w}))}, \qquad
\mathrm{eff}_A(\w \;|\; \bar{\w}) \eqdef \frac{\Phi_A(\M(\w))}{\Phi_A(\M(\bar{\w}))}.
\end{equation*}
Because $\Phi_D$ and $\Phi_A$ are positively homogeneous, the notion of efficiency as defined above can be interpreted in terms of the relative numbers of trials needed to achieve a given criterion value. For instance, $\mathrm{eff}_D(\w|\bar{\w})=0.99$ means that, asymptotically, to achieve the same amount of information measured by the criterion of $D$-optimality, we only need $99\%$ of the trials using the design $\bar{\w}$ compared to number of trials needed if we use the design $\w$. Let $\w \in \Xi_R$. If $\w^D$ is a $D$-optimal design and $\w^A$ is an $A$-optimal design, then (see \cite{Puk}, Ch.6)
\begin{equation}\label{eq:effbnd}
\mathrm{eff}_D(\w \;|\; \w^D) \geq \frac{m}{\max_{x \in \X} d_x(\w)}, \qquad
\mathrm{eff}_A(\w \;|\; \w^A) \geq \frac{\tr(\M^{-1}(\w))}{\max_{x \in \X} a_x(\w)}.
\end{equation}
If $\w$ is the current design and its efficiency compared to the optimal design is almost equal to $1$, there is no practical reason to continue the computation.

\subsection{Methods for computing optimal designs of experiments}

In this paper, we are concerned with the numerical computation of (approximate) $\Phi$-optimal designs as given in \eqref{eq:OPT}, with special focus on $D$- and $A$-optimality. Within the field of optimal experimental design, the first contributions to solving this problem were made by V. V. Fedorov and H. Wynn (e.g., \cite{Fedorov} and \cite{Wynn}), who developed the so-called vertex direction methods (VDMs) that are closely related to the Frank-Wolfe algorithm. In each iteration, VDMs move the current design $\w$ in the direction of $\e_x$ for some design point $x$ while decreasing all remaining components of $\w$ by the same proportion. Here, $x$ can be chosen to maximize $g_x(\w)$. Although these methods converge, some of them monotonically, they tend to be inconveniently slow. More efficient variants (e.g., \cite{Atwood73}) also allow the decrease of a single component of the current design. A related alternative approach called the vertex exchange method for $D$-optimality (VEM) was proposed by Bohning~\cite{Bohning}.
\bigskip

The VEM algorithm (see Algorithm~\ref{bex} for a pseudo-code) is one of the simplest special cases of the class of algorithms presented in this paper. In each iteration, with a current design $\w$, VEM selects a pair of points $k, l \in \X$ defined by 
\begin{equation*}
  k \in \argmin\{d_u(\w): u \in \supp(\w)\}, \: l \in \argmax\{d_v(\w): v \in \X\}.
\end{equation*}
We call such a pair $(k,l)$ a Bohning's pair of design points. Then, the VEM computes $\alpha^* \in [-w_l,w_k]$ such that $\Phi_D$ in $\M(\w+\alpha\e_l-\alpha\e_k)$ is maximized. Such an $\alpha^*$, which we call the Bohning's step, can be analytically computed. After each exchange, the variance function $\d$ is recomputed. 
\bigskip

\begin{algorithm}
	\SetKwInOut{Input}{input}
	\SetKwInOut{Output}{output}
	\SetKwRepeat{Do}{do}{while}
	\Input{Regressors $\f(1),\ldots,\f(n)$ of the model, required efficiency $\mathit{eff}$, maximum time $t.max$}
	\Output{Approximate design $\w$} 
	\BlankLine 
	Generate a regular $m$-point design $\w$;\\
	\While {$\mathit{eff.act}(\w) < \mathit{eff} \mathrm{ and } \: time<t.max$}
	   {Let $k$ belong to $\argmin\{d_u(\w): u\in\supp(\w)\}$\\
	    Let $l$ belong to $\argmax\{d_v(\w): v\in \X\}$\\
		Let $\alpha^*$ belong to $\argmax\{\Phi_D(\M(\w+\alpha\e_l-\alpha\e_k)): \alpha \in [-w_l, w_k]\}$\\
		$w_k \leftarrow w_k - \alpha^*$; $w_l \leftarrow w_l + \alpha^*$  	
	}
	\caption{A variant of the Bohning's vertex exchange method (VEM). The value of $\mathit{eff.act}(\w)$ is calculated as a lower bound on the efficiency of the current design $\w$ based on a standard formula (see \eqref{eq:effbnd}). The variable $time$ measures the time elapsed from the start of the computation.} \label{bex} 
\end{algorithm}

A substantially different strategy for solving the problem \eqref{eq:OPT}, involving simultaneously updating all components of $\w$, is implemented in the multiplicative algorithm (MUL). The MUL and its variants can be attributed to \cite{STT}; extensions can be found in \cite{DPZ} and \cite{Harman}.
\bigskip

The three methods of VDM, VEM and MUL were more recently combined by Yu (\cite{Yu}) in the ``cocktail'' algorithm for $D$-optimality. This approach often increases the speed while preserving convergence.
\bigskip

Another efficient recent method is presented in \cite{YBT}. It is the simplicial decomposition algorithm where the master non-linear problem is solved by a second-order Newton method and can be used for various twice-differentiable concave criteria, including both $D$- and $A$-optimality. See also \cite{Dariusz} for an earlier application of simplicial decomposition in the optimal design of experiments and further references.
\bigskip

The problem of searching for an optimal design can be also solved by specific mathematical programming methods. Namely, in \cite{VBW}, it is shown that the problem of $A$-optimality can be cast as semidefinite programming, while $D$-optimality can be reduced to maxdet programming. More recently, the work \cite{Sagnol} enables one to solve both $A$- and $D$-optimal design problems (for the approximate design case without additional constraints\footnote{For an improvement that also enables solving exact design problems and approximate design problems with non-standard constraints using the second-order cone programming, see \cite{SagnolHarman}.}) by the second-order cone programming approach.
\bigskip

The optimal design problem is closely related to the minimum volume enclosing ellipsoid (MVEE) problem. Considerable attention has been paid to developing fast algorithms for the MVEE problem in mathematical programming and optimization literature. For the latest developments, see \cite{Todd}, \cite{selin2}, and \cite{Todd-book}.
\bigskip

Consider a set $\{\f_1,\ldots,\f_n\}$ spanning $\R^m$. Then, the task of finding MVEE $\mathcal{E}(\M):=\{\f\in\R^m: \f'\M \f\leq 1\}$, $\M \in \mathcal{S}_+^m$, containing $\f_1,\ldots,\f_n$ can be cast as
\begin{equation*}
\left.\begin{array}{rl}
\min_{\M\in \mathcal{S}^m_{++}} & -\log\det\M \\
\hbox{subject to} & \f'_i\M \f_i\leq 1,\ i=1,\ldots,n
\end{array}\right.
\end{equation*}
with the dual problem
\begin{equation*}
\left.\begin{array}{rl}
\max_{\w\geq 0} & \log\det\sum_{i=1}^n w_i \f_i \f'_i \\
\hbox{subject to} & \sum_{i=1}^n w_i=1.
\end{array}\right.
\end{equation*}
Thus, the problem of finding the $D$-optimal design is equivalent to the MVEE problem (see, e.g., \cite{Todd-book}), making the task of developing a fast and efficient algorithm for computing $D$-optimal designs relevant for a wider optimization community.
\bigskip

With $D$-optimality being the most popular criterion, the literature focusing on the computational issues of $A$-optimality is scarcer, although in some areas, the use of $A$-optimality is very natural. In particular, the computation of the important $I$-optimal designs\footnote{These are occasionally called $IV$- or $V$-optimal designs.} (e.g., \cite{CookNachtsheim}, \cite{GJS}) can be converted into $A$-optimality (cf.  \cite{Atkinson}, Section 10.6). That is, $I$-optimal designs on a finite design space can also be computed using the algorithm developed in this paper. In addition to the mathematical programming methods mentioned above, a generalization of the vertex direction method that focuses on $A$-optimality, together with an analytical formula for the optimal VDM step-length, is presented in~\cite{selin}. 

\subsection{Summary of contributions}

In Section \ref{S:sam}, we introduce a general class of methods for finding optimal approximate designs that we call the subspace ascent method (SAM) and describe some known algorithms as special cases.
\bigskip

In Section \ref{S:rex}, we present our key method:  REX (randomized exchange method). REX is a simple batch-randomized exchange algorithm that is a member of the SAM family and whose performance is superior to the state-of-the-art algorithms in nearly all problems we tested. The proposed algorithm can be viewed as an efficient extension or combination of both the VEM algorithm and the KL exchange algorithm that is used to compute exact designs (see \cite{Atkinson}). In the same section, we compare our method to known subspace ascent/descent methods that were recently developed in the optimization and machine learning communities by highlighting the similarities and differences. The numerical comparison of the state-of-the art algorithms with this newly proposed method is the subject of Section \ref{S:ex}.  
\bigskip

In the Appendix, we provide a proof of convergence of the proposed algorithm in the case of $D$-optimality. Additionally, in the existing literature, many of the mentioned methods have only been tailored to $D$-optimality, and although the transition to $A$-optimality is possible, it is not trivial. Therefore, in the Appendix, we also present formulas for the optimum exchange of weights for $A$-optimality, which, to the best of our knowledge, were previously only derived for the $D$-optimality criterion.

\section{Subspace ascent method}\label{S:sam}

In Section~\ref{subsec:SAM}, we propose a generic algorithmic paradigm for solving \eqref{eq:OPT}. It helps illustrate several existing approaches and the newly proposed method from a broader perspective. We refer to this paradigm as the ``Subspace Ascent Method'' (SAM), formalized as Algorithm~\ref{alg:SAM}. Subsequently, in Section~\ref{sub:rand}, we briefly comment on the context of the ``subspace ascent'' idea within existing optimization, linear algebra and machine learning literature.

\subsection{SAM}\label{subsec:SAM}

SAM (Algorithm~\ref{alg:SAM}) is an iterative procedure for finding an approximate $\Phi$-optimal design. We start with an arbitrary regular design $\w^0\in \Xi$. In iteration $k$, a subset $S_k$ of the design points is chosen via a certain rule (e.g., a small random subset of all design points). In this iteration, we only allow for weights $w_x$ for $x\in S_k$  to change. All other weights are frozen at their last values. That is, we are deliberately reducing our search for $\w^{k+1}$ to a subset of the set of all designs $\Xi$ at the intersection of $\Xi$ and a particular affine subspace of $\R^n$:
\[\Xi_k \eqdef  \Xi \cap \cL_k,
\qquad  \cL_k \eqdef \{\w\;:\; w_x = w^k_x \text{ for all } x\notin S_k\},\]
Note that $\w^k \in \Xi_k$ for all $k$ and, therefore, $\max_{\w\in \Xi_k} \Phi(\M(\w)) \geq \Phi(\M(\w^k))$. In particular, there exists $\w^{k+1}\in \Xi_k$ such that the ascent condition $\Phi(\M(\w^{k+1})) \geq \Phi(\M(\w^k))$ is satisfied.

\begin{algorithm}
\vspace{0.5cm}
\SetKwInOut{Input}{input}
\SetKwInOut{Output}{output}
\Input{Initial regular design $\w^0\in \Xi$}
\Output{Approximate design $\w^k$}
Set $k \leftarrow 0$ \\
\While{$\w^k$ does not satisfy a stopping condition}
{Select a subset $S_k$ of the set of design points $\X=\{1,2,\dots,n\}$\\
 Define the \emph{active subspace} of $\Xi$ as $\Xi_k \eqdef \{\w\in \Xi : w_x = w^k_x \text{ for all } x\notin S_k\}$\\
Compute $\w^{k+1}$ as an approximate solution of $
\max_{\w \in \Xi_k} \Phi(\M(\w))$ satisfying the \emph{ascent condition} $\Phi(\M(\w^{k+1})) \geq \Phi(\M(\w^k))$\\
Set $k \leftarrow k+1$
}
\caption{Subspace Ascent Method (SAM)}
\label{alg:SAM}
\end{algorithm}

If $|S_k|=1$, then $\Xi_k$ is a singleton and the method cannot progress. As soon as $|S_k|> 1$, this problem is removed, and $\Xi_k$ has a chance to contain points better than $\w^k$.
\bigskip

SAM is a generic method in the sense that it does not specify how the set $S_k$ is chosen or how the approximate solution $\w^{k+1}$ is obtained. Formally, any optimum design algorithm for a discrete design space (including VDMs and the MUL) is a special case of SAM if we allow the choice $S_k=\X$. However, the intended philosophy of the SAM approach is to make the sets $S_k$ much smaller than $\X$. This makes the partial optimization problem small-dimensional and potentially simple to solve, at least approximately.
\bigskip

The algorithm VEM chooses $S_k$ to be a two point set, which allows the partial optimization problem to be analytically solved for $D$-optimality (and, as we show in the Appendix, for $A$-optimality). However, this method requires overly frequent updates of the set $S_k$, which is a computationally non-trivial operation for large $n$. Therefore, the speed of VEM significantly declines with the increase of the design space. 
\bigskip

Similarly, SAM also includes the algorithm YBT (\cite{YBT}), which chooses $S_k$ to be the support of the current design enriched by one additional support point $x$. Then, a partial optimization problem is solved by a specific constrained Newton method. However, this approach is generally inefficient for larger values of $m$ because then the sets $S_k$ tend to be large and the efficiency of the Newton method deteriorates rapidly as the dimension of the problem increases (as demonstrated in Section \ref{S:ex}).  
 
\section{Randomized exchange algorithm}\label{S:rex}

In this section, we propose a particular case of SAM for which we coin the name {\em randomized exchange algorithm (REX)}.  Then, we briefly comment on the connection of this method to the existing literature on mathematical optimization and machine learning.

\subsection{REX}

Let $\w$ be a regular design and $\g(\w)$ be an $n$-dimensional vector with components $g_x(\w)$, as defined in the introduction. Recall that the value $g_x(\w)$ is a $\w$-based estimate of the plausibility that $x$ is the support point of an optimal design.  
\bigskip

The general idea behind the proposed batch-randomized algorithm is to repeatedly select a batch of several pairs of design points (the number of pairs is not fixed and may vary from iteration to iteration) and randomly perform optimal exchanges of weights between the pairs of design points.  The selection of the batch depends on the vector $\g(\w)$ evaluated in the best available design $\w$. Its size can be fine-tuned by a tuning parameter $\gamma \geq 1/m$. We will call the resulting algorithm the randomized exchange algorithm (REX). We now proceed to a description of the REX algorithm (See Algorithm \ref{rex} for a concise pseudo-code). 

\begin{enumerate}
  \item {\bf LBE step.} Given the current design $\w$, compute $\g(\w)$ and subsequently perform the ``leading Bohning exchange'' (LBE) as follows:
  \begin{equation}\label{eqn:LBE}
   \w \leftarrow \w + \alpha^*_{k,l}(\w)(\e_l-\e_k),
  \end{equation}
where
  \begin{eqnarray*}
   k &\in& \argmin\{g_u(\w): u \in \supp(\w)\},\\
   l &\in& \argmax\{g_v(\w): v \in \X\},\\
   \alpha^*_{k,l}(\w) &\in& \argmax\{\Phi(\M(\w+\alpha(\e_l-\e_k))): \alpha \in [-w_l, w_k]\}
  \end{eqnarray*}
  An optimal step $\alpha^*_{k,l}(\w)$ is called ``nullifying'' if it is equal to either $-w_l$ or $w_k$. 
  \item {\bf Subspace selection.} Subspace $S$ of $\X$ in which the method will operate arises as the union of two sets. One is formed by a greedy process ($S_{\text{greedy}}$) informed by the elements of the vector $\g(\w)$, and the other is identical to the support of the current design $\w$ ($S_{\text{support}}$).
  \begin{enumerate}
  \item {\bf Greedy.} Set $L=\min(\gamma m,n)$, and choose points \[S_{\text{greedy}} = \{l^*_1,\ldots,l^*_L\} \subseteq \X,\] where $l^*_i$ corresponds to the $i$th largest component of the vector $\g(\w)$. 
  \item {\bf Support.} Set \[S_{\text{support}} = \supp(\w).\] Let $K$ be the size of $\supp(\w)$: $K=|\supp(\w)|$. 
  \item {\bf Active subspace.} The active set of design points is defined as
  \[S = S_{\text{greedy}} \cup S_{\text{support}}.\]
  The weights of no other design points are modified in this iteration. Note that if $m^2\geq n$ or if $K = n$, then $S=\X=\{1,2,\dots,n\}$. However, a standard situation is $|S|<n$, and our method operates in a proper subspace of the full design space $\X$.
 \end{enumerate}
 
\item {\bf Subspace step.} We now perform a specific update of the design points in $S$. That is, we allow $w_v$ for $v \in S$ to be updated. Elements $w_v$ for $v\notin S$ are kept unchanged.

\begin{enumerate}
\item {\bf Formation of pairs.} Let $(k_1,\ldots,k_K)$ be a uniform random permutation of $S_{\text{support}}$, and let $(l_1,\ldots, l_L)$ be a uniform random permutation of $S_{\text{greedy}}$. Form the sequence
\begin{equation}(k_1,l_1), (k_2,l_1),\dots, (k_K, l_1), \dots, (k_1, l_L), (k_2, l_L), \dots, (k_K,l_L) \label{eq:rand_perm}\end{equation} of $K\times L$ pairs of active design points.

  \item {\bf Update.} If the LBE step was nullifying, sequentially perform all nullifying $\Phi$-optimal exchanges between the $K\times L$ pairs in \eqref{eq:rand_perm} with the corresponding updates of $\w$ and $\M(\w)$.  If the LBE was {\em not} nullifying in the current iteration, sequentially perform all $\Phi$-optimal  exchanges between the $K\times L$ pairs in \eqref{eq:rand_perm} with the corresponding updates of $\w$ and $\M(\w)$.
 
  \end{enumerate}  
  
   \item {\bf Stopping rule.} If a stopping rule is satisfied, stop. Otherwise, iterate by returning to step~1.
\end{enumerate}

We remark that for the criteria of $D$- and $A$-optimality, we have rapid explicit formulas providing the optimal exchanges in steps 1 and 3b of REX; see Appendix~\ref{appendix}. However, in principle, REX can also be applied to any other concave (even non-differentiable) optimality criterion using numerical procedures for the one-dimensional convex optimization on an interval, provided that we find a suitable analogy to the function $\g(\cdot)$.
\bigskip

REX makes explicit use of the fast optimal weight exchange between two design points, which makes it much less dependent on $m$ than YBT. Moreover, REX requires less frequent re-computations of the function $\g(\w)$, which leads to significant computational savings compared to VEM for large $n$. Note that REX is very simple. In particular, it is simpler than the hybrid cocktail algorithm of Yu \cite{Yu}. Moreover, unlike the method of \cite{Yu}, the result does not depend on the ordering of the design points. Finally, as we demonstrate, in most cases, it is also numerically more efficient than both of the state-of-the art methods proposed in \cite{Yu} and \cite{YBT}. 

\begin{algorithm}[t!]
	\SetKwInOut{Input}{input}
	\SetKwInOut{Output}{output}
	\SetKwRepeat{Do}{do}{while}
	\Input{Regressors $\f(1),\ldots,\f(n) \in \R^m$, parameter $\gamma \geq 1/m$, threshold efficiency $\mathit{eff}$, maximum time $t.max$}
	\Output{Approximate design $\w$} 
	\BlankLine 
	Generate a random regular design $\w$\\
	\While {$\mathit{eff.act}(\w) < \mathit{eff}$ and $time<\mathit{t.max}$}
	{Perform the LBE in $\w$ as given by \eqref{eqn:LBE}.\\
	Let $\k$ be the vector corresponding to a random permutation of the $K$ elements of $\supp(\w)$\\
	Let $\l$ be the vector corresponding to a random permutation of the indices of the $L=\min(\gamma m,n)$ greatest elements of $\g(\w)$\\
	\For {$l$ in $1:L$}
		{\For{$k$ in $1:K$}
			{Find $\alpha^*$ in $\argmax\{\Phi(\M(\w+\alpha(\e_{\l_l}-\e_{\k_k}))): \alpha \in [-w_{\l_l}, w_{\k_k}]\}$\\
		     \If{the LBE was not nullifying or $\alpha^*=-w_{\l_l}$ or $\alpha^*=w_{\k_k}$} {
			   $w_{\k_k} \leftarrow w_{\k_k} - \alpha^*$ 
			   $w_{\l_l} \leftarrow w_{\l_l} + \alpha^*$
			   }
		    }
		}
	}
	\caption{Randomized exchange algorithm (REX). The value of $\mathit{eff.act}(\w)$ is the lower bound on the current design efficiency compared to the optimal design. For $D$- or $A$-optimality, it is possible to use formulas \eqref{eq:effbnd}. The variable $time$ determines the time from the beginning of the computation. The optimal step-length $\alpha^*$ can be computed by explicit formulas given in the appendix (in the case of $D$- or $A$-optimality) or by a procedure for the one-dimensional convex optimization on a finite interval.}\label{rex} 
\end{algorithm}

\subsection{Connection to existing literature on subspace descent methods} \label{sub:rand}

Subspace ascent/descent methods of various types have recently been extensively studied in the optimization \cite{Nes12, PCDM, APPROX}, linear algebra \cite{LL10, LeeSid13, GR-SIMAX}, machine learning \cite{SDCA,  Quartz, mS2GD} and image processing \cite{SCP} literature.  Of particular importance are the {\em randomized} and {\em greedy} variants.

 Randomized subspace descent methods operate by taking steps in {\em random} subspaces of the domain of interest according to some  distribution over subspaces (all of which typically have a relatively small and fixed dimension) that is fixed throughout the iterative process. Once a subspace has been identified, a gradient, Newton or a proximal point step is usually taken in the subspace. These methods have been found to scale to extremely large dimensions for certain applications \cite{PCDM, HYDRA} (e.g., billions of unknowns), and have become the state of the art for a variety of big data analysis tasks.  Randomized selection rules are, in practice, often better than cyclic rules that pass through the subspaces in a fixed predefined order. Intuitively, this happens because randomization helps avoid/break unfavorable orderings. It was recently shown by Sun and Ye~\cite{Ye2016} that there is a large  theoretical gap between randomized and cyclic rules (proportional to the square of the dimension of the problem), which favors randomized rules. 

Greedy subspace descent methods operate by taking steps in subspaces selected {\em greedily} according to some potential/importance function of interest. Again, once a subspace has been identified, typically a gradient, Newton or a proximal point step is taken in the subspace.  The ideal greedy method employs the {\em maximum improvement} rule, which requires that at any given iteration, one should choose the subspace that leads to the maximum improvement in the objective. Since this is almost always impractical, one needs to resort to approximating the maximum improvement rule by an efficiently implementable proxy rule.  Whether one should use a greedy or a randomized method depends on the problem structure, as both have applications in which they dominate. 

Let us now highlight some of the key similarities and differences between REX and existing subspace descent methods:
\begin{enumerate}

\item {\bf Support plays a role.} The subspace selection rule in REX has a {\em greedy} component through the inclusion of the set $S_{\text{greedy}}$. However, unlike most existing subspace descent methods, it also has an ``active set / support'' component through the inclusion of $S_{\text{support}}$.

\item {\bf Greedy subspaces.} Greedy subspace rules beyond subspaces of dimension one have not been explored in the optimization and machine learning literature. To the best of our knowledge, the only exception is the work of Csiba and Richt\'{a}rik~\cite{CsibaPL}. However, both their greedy subspace rule and the problem that they tackle is different from ours.

\item {\bf Subspace step.} While existing subspace descent methods typically rely on traditional deterministic (as opposed to stochastic) optimization steps such as gradient or Newton  steps, the REX performs a {\em randomized} pairwise exchange step in the subspace. We perform pairwise exchange steps as for the most important optimal design criteria. These can be calculated exactly in closed-form, which facilitates fast computations. This is similar to the sequential minimal optimization technique of Platt \cite{plattSMO} that is influential in the machine learning literature. Methods with randomized steps performed in the subspace, such as REX, are rare in the optimization literature. One of the earliest examples of such an approach is the CoCoA framework of Jaggi et al.\ \cite{CoCoA}, aimed at distributed primal-dual optimization in machine learning. Both their subspace subroutine and the problem that they tackle are very different from ours.

\item {\bf Random reshuffling.} Randomization in REX is present in the form of a {\em random permutation} of pairs from $S_{\text{support}}\times S_{\text{greedy}}$. Methods relying on random permutations can be seen as a hybrid between stochastic and cyclic methods. These approaches are not theoretically understood due the intrinsic difficulty in capturing their behaviors using complexity analysis.  The work of G\"{u}rb\"{u}zbalaban et al.\ \cite{randomshuffle} provides one of the first insights into this problem. However, their techniques do not apply to our problem or algorithm. 

\item {\bf Nonseparable constraints.} Virtually all modern stochastic and greedy subspace descent methods in optimization apply to unconstrained problems\footnote{Problems with a convex constraint admitting a fast (i.e., closed form) projection are considered equally simple as unconstrained problems.}, or to problems with separable constraints structure, as this property is crucial in the analysis of these methods. One of the early works that explicitly tackles non-separable constraints is that of Necoara et al. \cite{Necoara-constraints}. However, their work applies to linear constraints only. 

\item {\bf Smoothness.} Vast majority of papers on stochastic and greedy subspace descent methods assume that the objective function has a Lipschitz continuous gradient (this property is often called $L$-smoothness in the machine learning literature). However, this is not true in our case. Hence, fundamentally different convergence analysis techniques are needed in our case. Only a handful of subspace descent methods do not rely on this assumption. One of them is the stochastic Chambolle-Pock algorithm \cite{SCP}, which is a stochastic extension of a famous state-of-the-art method in the area of computational imaging: the Chambolle-Pock method.

\end{enumerate}

\section{Numerical comparisons}\label{S:ex}

In this section, we compare the performance of an \texttt{R} (\cite{R}) implementation of REX and  two state-of-the-art algorithms for computing the $D$-optimal design of experiments. As a first candidate,
we used the cocktail algorithm (CO), more precisely its \texttt{R} implementation, which was made available by the author of CO (see \cite{Yu}).\footnote{There are two variants of CO. One has a pre-specified neighborhood structure, and the other has nearest neighbors computed on the fly. For each model, we selected the variant that performed better.} As a second algorithm, we selected the Newton-type method \cite{YBT} (YBT). For a fair comparison, the \texttt{SAS} code kindly provided by the authors of YBT was converted into \texttt{R} with as high fidelity as possible. The \texttt{R} codes of REX are available at
\bigskip

\hyperlink{http://www.iam.fmph.uniba.sk/ospm/Harman/design/}{http://www.iam.fmph.uniba.sk/ospm/Harman/design/}.
\bigskip

Note that it is possible to directly use the codes without any commercial software and without technical knowledge of the algorithms. 
\bigskip

In the following two subsections, we compare the competing algorithms for two very different, generic $D$-optimum design problems with widely varying sizes. This approach provides comprehensive information about the relative strengths of the methods in a broad range of situations. However, we remark that the numerical comparisons must be taken with a grain of salt because they depend on a multitude of factors. Different implementations and different hardware capabilities can lead to somewhat different relative results. 
\bigskip

In all computations with REX, we chose the tuning constant $\gamma=4$, uniformly random $m$-point initial designs. We used a computer with a 64-bit Windows 10 operating system running an Intel Core i7-5500U CPU processor at $2.40$ GHz with $8$ GB of RAM.

\subsection{Quadratic models}

Consider the full quadratic regression model on the $d$-dimensional cube $[-1,1]^d$, $d \in \mathbb{N}$. In other words, for $(t_1,\ldots,t_d)' \in [-1,1]^d$ the observations are modeled as
\begin{equation}\label{quadmod}
Y(t_1,\ldots,t_d)=\beta_1 + \sum_{i=1}^d \beta_{i+1} t_i +\sum_{j=1}^{d} \sum_{k=j}^d \tilde{\beta}_{i,j} t_it_j + \varepsilon(t_1,\ldots,t_d),
\end{equation}
where $\tilde{\beta}_{1,1},\tilde{\beta}_{1,2},\ldots,\tilde{\beta}_{d,d}$ correspond to the parameters $\beta_{d+1},\beta_{d+2},\ldots,\beta_m$, $m=(d+1)(d+2)/2$. We discretized the cube $[-1,1]^d$ to obtain a lattice $\mathcal{D}$ of $n$ points, which are numbered by indices from $\X=\{1,\ldots,n\}$. Thus, the set of regressors $\{\f(1),\ldots,\f(n)\}$ corresponds to $\{(1,t_1,t_2,\ldots, t_{d-1}t_d, t_d^2)': (t_1,\ldots,t_d)' \in \mathcal{D}\}$. These are representatives of structured models, such as those used in the response surface methodology (see, e.g., \cite{rsm}). Each of the above mentioned algorithms was run for various combinations of values $n$ and $d$ (see Figure \ref{cube}).
\bigskip

\begin{figure}[!ht]
\begin{center}
	\includegraphics[width=14cm]{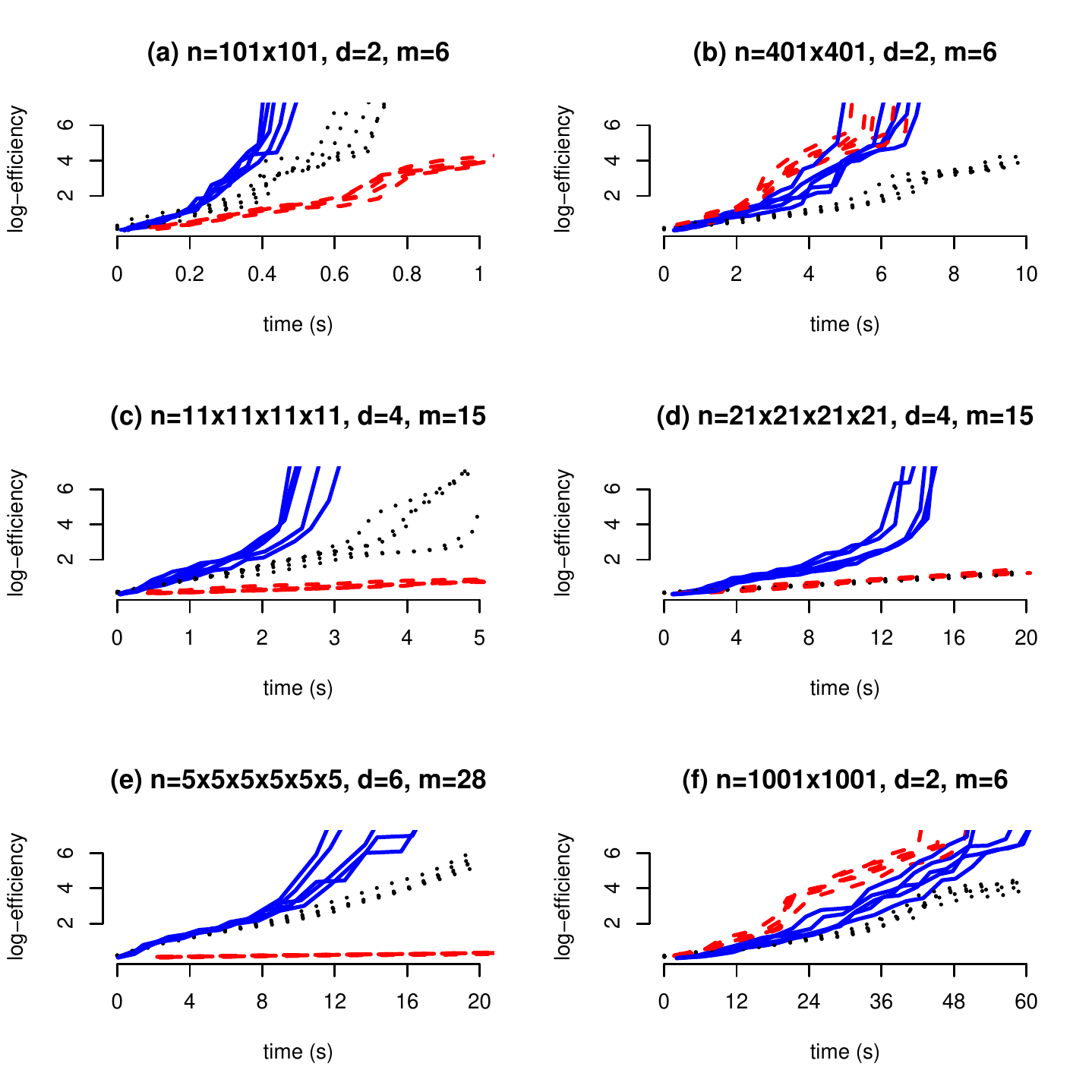}
\end{center}
	\caption{Typical performance of the algorithms CO (black dotted line), YBT (red dashed line) and REX (blue solid line) for $D$-optimality in the case of quadratic models \eqref{quadmod} on an $n$-point discretization of $[-1,1]^d$. The vertical axis denotes the log-efficiency $-\log_{10}(1-\eff)$, where $\eff$ is the $D$-efficiency of design. Thus, log-efficiencies $2$, $4$ and $6$ correspond to $D$-efficiencies $0.99$, $0.9999$ and $0.999999$, respectively. The horizontal axis corresponds to the computation time in seconds. Each algorithm was run $5$ times to illustrate the degree of variability.}\label{cube}
\end{figure}

The numerical results confirm the claim of \cite{YBT} that YBT is superior to CO for a large size $n$ of the design space and a small number $m$ of model parameters (see panels (b) and (f) of Figure \ref{cube}). However, for a smaller design space, particularly if the number of parameters is large, the algorithm CO tends to perform better than YBT (cf. panels (a), (c) and (e) of Figure \ref{cube}). A similar observation can also be drawn for other models, such as the one presented in the next subsection.
\bigskip

Importantly, the performance of REX is comparable or superior to both state-of-the-art algorithms for all combinations of $n$ and $m$ that we explored. 

\subsection{Random models}
As a second example, we chose a very different problem in which the regressors $\f(1),\ldots,\f(n)$ are drawn independently and randomly from the $m$-dimensional normal distribution $N_m(\0,\I)$.  These models represent unstructured and unordered optimal design situations in which the possible covariates of the model are drawn from a random population. The results are demonstrated in Figure \ref{random}.
\bigskip

\begin{figure}[!ht]
\begin{center}
	\includegraphics[width=14cm]{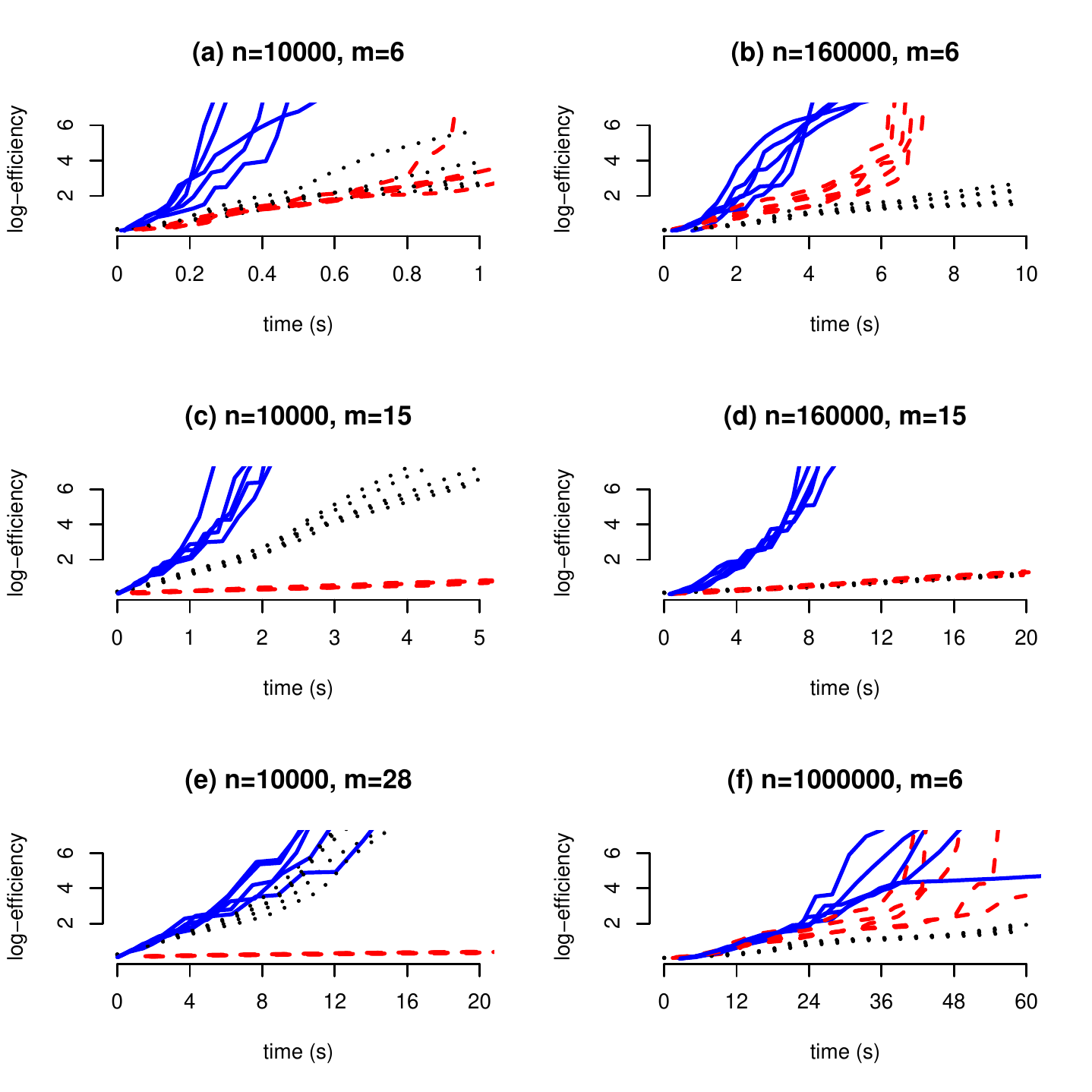}
\end{center}
	\caption{Typical performance of the algorithm CO (black dotted line), YBT (red dashed line) and REX (blue solid line) for $D$-optimality in the case of models with $n$ Gaussian random regressors of dimension $m$.  The vertical axis denotes the log-efficiency $-\log_{10}(1-\eff)$, where $\eff$ is the $D$-efficiency. The horizontal axis corresponds to the computation time. Each algorithm was run $5$ times.} 
	\label{random}
\end{figure}

The random models again show the superior performance of the proposed algorithm. The differences in favor of REX are even more pronounced compared to the quadratic models. This is possibly caused by the fact that on the unstructured set of regressors, procedures such as CO that depend on the ordering of design points are disadvantaged. Note that the algorithm REX is independent of the ordering of the design points.

\subsection{Further remarks on the numerical comparisons}

\textbf{$A$-optimality.} In addition to $D$-optimality, we performed a comparison of CO, YBT and REX for several problems under the criterion of $A$-optimality. Note that the algorithm CO was originally developed for $D$-optimality only, but the formulas given in the Appendix (Subsection \ref{Astep}) allowed us to modify it for the computation of $A$-optimal designs. For $A$-optimality, the relative efficiencies of the three algorithms were similar to the case of $D$-optimality, although the advantage of REX is generally less pronounced.
\bigskip

\noindent \textbf{Other algorithms for $D$-optimality.} For some models with a small number of design points and a very large number of parameters, the classical methods, particularly the multiplicative algorithm, outperform all three ``state-of-the-art'' algorithms (REX, YBT, CA), at least in an initial part of the computation. However, in a vast majority of test problems, the classical methods cannot compete with the three modern algorithms. A comprehensive comparison of all available methods is beyond the scope of this paper.   
\bigskip

\noindent \textbf{Disadvantages of REX.} We believe it to be important that the reader is informed of the disadvantages of our method we are aware of. First, the method occasionally generates streaks of only slowly increasing criterion values that are, however, usually followed by a rapid improvement. Second, we were not able to prove the convergence to the optimum\footnote{Note that the convergence of the criterion values per-se follows trivially from the monotonicity of REX. However, the convergence to the value \emph{optimal} for the problem at hand seems to be difficult to prove for a general concave criterion.} except for the case of $D$-optimality. Nevertheless, we also tested REX for other criteria, and the algorithm converged in all test problems.

\section{Conclusions}

As we have numerically demonstrated, the randomized exchange algorithm (REX) largely outperforms recent state-of-the-art methods for computing $D$-optimal approximate designs of experiments, mostly in the cases where the model consists of randomly generated or otherwise unstructured regressors. Overall, in comparison to the competing methods, the performance of REX deteriorates much less when both the size of the design space and the number of parameters increase. It is also worth noting that REX is concise, has relatively low memory requirements, and does not utilize any advanced mathematical programming solvers. In addition, for $D$-optimality, we theoretically proved the convergence of REX to the correct optimum. To adapt the proposed algorithm to the $A$-optimality criterion, we derived appropriate vertex-exchange formulas. 
\bigskip

REX offers many possibilities for further development. Besides the exploration of its convergence properties for a general concave criterion, there is a space for improvement in an optimized, perhaps adaptive selection of the parameter $\gamma$ that regulates the batch size. One way to further increase the speed of the algorithm is to incorporate specific methods of initial design construction and the deletion rules of non-supporting points (see \cite{del} for $D$-optimality and \cite{Pronzato} for $A$-optimality).
\bigskip

Moreover, it is possible to employ REX for the computations of optimal designs on continuous (infinite) design spaces. One approach is to choose a large number of random points inside the continuous design space, use the speed of REX for unstructured design spaces to construct the optimal approximate design on the finite sub-sample, and fine-tune the design by using standard routines of constrained continuous optimization. However, a numerical and theoretical analysis of this methodology is beyond the scope of this paper.
\bigskip

\textbf{Acknowledgments} The work of the first two authors was supported by Grant Number 1/0521/16 from the Slovak Scientific Grant Agency (VEGA). The last author acknowledges support through the KAUST baseline research funding scheme.


\appendix

\section{Optimal Step-length and Convergence}\label{appendix}

Let $\w \in \Xi$, $u,v \in \X$, and let $\Phi: \mathcal{S}^m_+ \to [0,\infty)$ be a criterion of design optimality. The \emph{optimal step-length} of weight exchange in the design $\w$ between design points $u$ and $v$ is any
\begin{equation}\label{exchange}
  \alpha^*_{u,v}(\w) \in \argmax\{\Phi(\M[\w+\alpha (\e_v-\e_u)]):\alpha \in [-w_v,w_u]\}.
\end{equation}
In accord with the definition \eqref{exchange}, we provide an optimal step-length formula for any design $\w \in \Xi_R$. We also prove the convergence of the REX algorithm for $D$-optimality.

\subsection{D-optimality: Optimal Step-length and a Proof of Convergence of REX}

Let $\w \in \Xi_R$ and $u,v \in \X$. If $w_u=w_v=0$, the optimal step-length is trivially $\alpha^*_{u,v}(\w)=0$. Therefore, we can assume that at least one of the weights $w_u$ and $w_v$ is strictly positive. We use the notation
\begin{equation*}
  d_{u,v}(\w):=\f'(u)\M^{-1}(\w)\f(v).
\end{equation*}
For $D$-optimality, it has been shown (see \cite{Bohning}, \cite{Yu}) that an optimal choice of the step-length in \eqref{exchange} is as follows. 
\bigskip

If $\f(u)$ and $\f(v)$ are linearly independent, then\footnote{Note that if $\f(u)$ and $\f(v)$ are linearly independent, the Cauchy-Schwarz inequality implies $d_u(\w)d_v(\w) > d^2_{u,v}(\w)$.}
\begin{equation}\label{eq:DoptBstepA}
  \alpha^*_{u,v}(\w)=\min\left\{w_u,\max\left\{-w_v,\frac{d_v(\w)-d_u(\w)}{2[d_u(\w)d_v(\w)-d^2_{u,v}(\w)]}\right\}\right\}.
\end{equation}
If $\f(u)$ and $\f(v)$ are linearly dependent, we can\footnote{For $d_u(\w) = d_v(\w)$, the choice of the optimal $\alpha^*_{u,v}(\w)$ is not unique.} set
\begin{equation}\label{eq:DoptBstepB}
\alpha^*_{u,v}(\w)=\begin{cases}
w_u & \mbox{ if } d_u(\w) < d_v(\w),\\
0 & \mbox{ if } d_u(\w) = d_v(\w),\\
-w_v & \mbox{ if } d_u(\w) > d_v(\w).
\end{cases}
\end{equation}
A $D$-optimal vertex-exchange step for $\w$ between design points $u,v$ is then
\begin{equation*}
  T^D_{u,v}(\w):=\w+\alpha^*_{u,v}(\w)(\e_v-\e_u).
\end{equation*}
The $D$-optimal vertex exchange step is called nullifying if $\alpha^*_{u,v}(\w)=w_u$ or if $\alpha^*_{u,v}(\w)=-w_v$, that is, if the $u$-th or the $v$-th component of $T^D_{u,v}(\w)$ is equal to $0$. The pair of indices
\begin{eqnarray*}
  u^*(\w) &\in & \argmin\{d_u(\w): u \in \supp(\w)\},\\
  v^*(\w) &\in & \argmax\{d_v(\w): v \in \X\} 
\end{eqnarray*}
will be called the $D$-optimal Bohning's vertex pair and the steplength $\alpha^*_{u,v}(\w)$ for $u=u^*(\w)$, $v=v^*(\w)$, will be called the $D$-optimal Bohning's step. In the following, we will use the well-known facts (i) $d_{u^*(\w)}(\w) \leq m$, (ii) $d_{v^*(\w)}(\w) \geq m/\eff_D(\w)$, and (iii) $\w$ is $D$-optimal if and only if $d_{v^*(\w)}(\w)=m$ (see, e.g., \cite{Atkinson}, 9.2). In (ii), $\eff_D(\w)$ denotes the $D$-efficiency of $\w$ relative to a $D$-optimal design.  

\begin{lemma}\label{lem:Dconv1} 
Let $u,v \in \X$ and $\w \in \Xi_R$. Then 
\begin{eqnarray}
&&\det(\M(\w)) \leq \det(\M[T^D_{u,v}(\w)]), \text{ and } \label{i}\\
&&\big\{ T^D_{u,v}(\w) \neq \w \} \Rightarrow \{ \det(\M(\w))<\det(\M[T^D_{u,v}(\w)]) \big\}. \label{ii}
\end{eqnarray}
Moreover, if $\w$ is not $D$-optimal and $u=u^*(\w)$, $v=v^*(\w)$ is the $D$-optimal Bohning's vertex pair, then $\det(\M(\w)) < \det(\M[T^D_{u,v}(\w)])$.
\end{lemma}
\begin{proof}
Inequality \eqref{i} follows from the choice of $\alpha^*_{u,v}(\w)$. We will prove \eqref{ii}. If $T^D_{u,v}(\w) \neq \w$, then $\alpha^*_{u,v}(\w) \neq 0$. Assume that $\alpha^*_{u,v}(\w)>0$. Then the rules for $\alpha^*_{u,v}(\w)$ imply $d_u(\w) < d_v(\w)$ and $w_u>0$. Now, to show that $\det(\M(\w)) < \det(\M[T^D_{u,v}(\w)])$, it is enough to verify that the directional derivative of $\log\det(\cdot)$ in $\M(\w)$ in the direction of $ \M[T^D_{u,v}(\w)]$ is strictly positive. Noting that $\M^{-1}$ is the gradient of $\log\det(\cdot)$ in a non-singular $\M$, we obtain that the directional derivative is
\begin{eqnarray*}
  \partial \log\det(\M(\w), \M[T^D_{u,v}(\w)])
  &=& \tr(\M^{-1}(\w)(\M[T^D_{u,v}(\w)]-\M[\w]))\\
  &=&\tr(\M^{-1}(\w)[\alpha^*_{u,v}(\w)(\f(v)\f'(v)-\f(u)\f'(u))])\\
  &=&\alpha^*_{u,v}(\w)(d_{v}(\w)-d_{u}(\w))>0.
\end{eqnarray*}
For the case $\alpha^*_{u,v}(\w) < 0$, the strict inequality in \eqref{ii} can be proved analogously.
 
Finally, we prove the last statement of the theorem. If $\w \in \Xi_R$ is not optimal, then $d_{v^*(\w)}(\w)>m$, i.e., (i) implies that $d_{u^*(\w)}(\w)<d_{v^*(\w)}(\w)$. However, $u^*(\w) \in \supp(\w)$, therefore $w_{u^*(\w)}>0$, and the formulas for the optimal step-length imply $\alpha^*_{u,v}(\w)>0$, i.e., $T^D_{u,v}(\w) \neq \w$. We can close the proof by using \eqref{ii}. 
\end{proof}

Recall that for $\w \in \Xi_R$ a Bohning's pair $u=u^*(\w),v=v^*(\w) \in \X$ is called nullifying if $\alpha^*_{u,v}(\w)$ is equal to $-w_v$ or $w_u$. 

\begin{lemma}\label{lem:Dconv2}
  For any $0<\epsilon<1$, there exists $K_\epsilon<\infty$ such that for all $\w \in \Xi_R$ satisfying $\eff_D(\w) \in [\epsilon,1)$ and a non-nullifying $D$-optimal Bohning's pair $u=u^*(\w)$, $v=v^*(\w)$, we have 
 \begin{equation*}
   \eff_D\left(T^D_{u,v}(\w) \right) \geq \eff_D(\w)\left[1+\frac{(\eff^{-1}_D(\w)-1)^2}{K_\epsilon}\right]^{1/m}.
 \end{equation*}  
\end{lemma}
\begin{proof}
Let $\tilde{u},\tilde{v} \in \X$ and $D_{\tilde{u},\tilde{v}}(\w):=d_{\tilde{u}}(\w)d_{\tilde{v}}(\w)-d^2_{\tilde{u},\tilde{v}}(\w)$. Clearly, $4m^{-2}D_{\tilde{u},\tilde{v}}(\w)$ is a continuous function of $\w$ on the compact $\Xi_\epsilon:=\{\w \in \Xi: \eff_D(\w) \in [\epsilon, 1]\} \subset\Xi_R$. Therefore, it is bounded on $\Xi_\epsilon$ from above by some constant $K_{\tilde{u},\tilde{v},\epsilon}<\infty$. Set $K_\epsilon:=\max_{\tilde{u},\tilde{v} \in \X} K_{\tilde{u},\tilde{v},\epsilon}$, which is clearly finite.

Assume first that $\f(u)$ and $\f(v)$ are linearly dependent. Then, the Bohning's step can only be non-nullifying if $d_u(\w)=d_v(\w)$. However, this cannot happen because $d_u(\w) \leq m$ from fact (i) and since $\eff_D(\w)<1$, $d_v(\w)>m$ from fact (iii).

Therefore, $\f(u)$ and $\f(v)$ are linearly independent. Then, since the $D$-optimal Bohning's step given by \eqref{eq:DoptBstepA} is assumed to be non-nullifying, we have $\alpha^*_{u,v}(\w)=[d_v(\w)-d_u(\w)]/[2D_{u,v}(\w)]$. From the matrix determinant lemma and the Sherman-Morrison formula, it is straightforward to show that
for any positive definite $m \times m$ matrix $\M$, any $\alpha \in \R$ and $\x, \y \in \R^m$ 
\begin{eqnarray}\label{eq:detchange}
&&\det(\M+\alpha(\y\y'-\x\x'))=\nonumber \\
&&\det(\M)\left[(1+\alpha\y'\M^{-1}\y)(1-\alpha\x'\M^{-1}\x) + \alpha^2(\y'\M^{-1}\x)^2 \right].
\end{eqnarray} 
Setting $\alpha=\alpha^*_{u,v}(\w)$, $\y=\f(v)$, $\x=\f(u)$ and $\M=\M(\w)$ in \eqref{eq:detchange} gives
\begin{equation*}
  \det(\M(T^D_{u,v}(\w)))=\det(\M(\w))\left(1+\frac{[d_v(\w)-d_u(\w)]^2}{4D_{u,v}(\w)}\right).
\end{equation*}  
  Recalling that $4D_{u,v}(\w)<m^2K_{\epsilon}$, $d_u(\w) \leq m$ and $d_v(\w) \geq m/\eff_D(\w)$, we obtain the inequality from the lemma.
\end{proof}

\begin{theorem}
  The algorithm REX with the step-length rule defined by \eqref{eq:DoptBstepA} and \eqref{eq:DoptBstepB} converges to a $D$-optimal design in the sense that the sequence of $D$-criterion values of designs produced by the algorithm converges to the $D$-optimal value $\Phi^*$.\footnote{The convergence is the sure convergence of random variables.}\footnote{The $D$-optimal value $\Phi^*$ is defined as $\Phi_D(\M^*)$, where $\M^*$ is the $D$-optimal information matrix.} 
\end{theorem}
\begin{proof}
Let the initial design of REX be $\w^0 \in \Xi_R$, let $\epsilon := \eff_D(\w^0)$, and let $(\w^i)_{i=1}^\infty$ be the sequence of designs that REX generates by the leading $D$-optimal Bohning's exchanges at iterations $1,2,\ldots$. Because all the transformations used by REX are non-decreasing with respect to $\Phi_D$, the sequence $(\Phi_D(\M(\w^i)))_{i=1}^\infty$ has a limit, say $\Phi^+$. Assume that $\Phi^+$ is not equal to the $D$-optimal value $\Phi^*$ of the problem and let $\eff_+:=\Phi^+/\Phi^* <1$. Now, inspired by the proof strategy of Bohning (\cite{Bohning}), we will split the proof into two cases.\footnote{We will deal with Case 1 differently than Bohning.} In each case, we show a contradiction, which then implies the claim of the theorem.

Case 1: Assume that there is only a finite number of non-nullifying leading $D$-optimal Bohning's exchanges. Each nullifying exchange can only decrease the size of the support of the current design, or it can keep the size of the support constant. Since there can only be a finite number of support-reducing nullifying exchanges, there is some index $i$ such that all the exchanges performed after the $i$-th iteration are nullifying and keep a constant size of the support. Hence, after the $i$-th iteration, the algorithm does not alter the \emph{set} of all non-zero weights. At the same time, REX is designed such that after the $i$-th iteration, it will perform only nullifying exchanges (also in the non-leading exchanges). Since $\X$ is finite, this means that the generated designs must return to one of the previous designs after a finite number of steps. This is impossible because according to Lemma \ref{lem:Dconv1}, each $D$-optimal Bohning's step (even if it is nullifying) strictly increases the criterion value of a sub-$D$-optimal design.

Case 2: Assume that there is an infinite number of non-nullifying leading $D$-optimal Bohning's steps, at iterations $n_1, n_2, \ldots$. Then, Lemmas \ref{lem:Dconv1} and \ref{lem:Dconv2} imply that
 \begin{equation}\label{eq:divergence}
   \eff_D(\w^{n_k}) \geq \epsilon \left[1+\frac{(\eff_+^{-1}-1)^2}{K_\epsilon}\right]^{k/m}
 \end{equation}    
for all $k=1,2,\ldots$ and some positive real numbers $\epsilon, K_\epsilon$. This is impossible if we assume that $\eff_+<1$ because the RHS of  \eqref{eq:divergence} converges to infinity with $k \to \infty$ but the efficiency of any design compared to the $D$-optimal design is bounded by $1$ from above.
\end{proof}

\subsection{A-optimality: Optimal Step-length}\label{Astep}

Let $u,v \in \X$ be a fixed pair of design points and let $\w \in \Xi_R$ be a fixed design. Let at least one of the weights $w_u$ and $w_v$ be strictly positive. 
\bigskip

Denote $\f_u :=\f(u)$, $\f_v :=\f(v)$, $\M:=\M(\w)$ and $\V:=\M^{-1}(\w)$. For $\alpha$ in the open interval $I=(-w_v,w_u)$, the matrix $\M+\alpha \f_v\f'_v$ as well as the information matrix
\begin{equation*}
  \M_\alpha:=\M+\alpha \f_v\f'_v-\alpha\f_u\f'_u
\end{equation*}
are positive definite, and we can use the Woodbury formula to obtain\footnote{Note that \eqref{eq:detchange} implies $1+\alpha \f'_v\V\f_v-\alpha \f'_u\V\f_u-\alpha^2 \f'_u\V\f_u \f'_v\V\f_v+\alpha^2 (\f'_u\V\f_v)^2$ is strictly positive for $\alpha \in I$.}
\begin{equation*}
\M^{-1}_\alpha=\V+
\frac{\V[\alpha\f_u\f'_u-\alpha \f_v\f'_v-\alpha^2 \f'_u\V\f_v (\f_u\f'_v+\f_v\f'_u)+\alpha^2 (\f'_v\V\f_v\f_u\f'_u+\f'_u\V\f_u\f_v\f'_v)]\V}{1+\alpha \f'_v\V\f_v-\alpha \f'_u\V\f_u-\alpha^2 \f'_u\V\f_u \f'_v\V\f_v+\alpha^2 (\f'_u\V\f_v)^2 },
\end{equation*}
\begin{equation*}
-\tr\M^{-1}_\alpha=-\tr\V-\alpha\frac{\f'_u\V^2\f_u-\f'_v\V^2\f_v- 2\alpha \f'_u\V\f_v\f'_u\V^2\f_v+\alpha \f'_v\V\f_v\f'_u\V^2\f_u+\alpha \f'_u\V\f_u\f'_v\V^2\f_v}
{1+\alpha \f'_v\V\f_v-\alpha \f'_u\V\f_u-\alpha^2\f'_u\V\f_u \f'_v\V\f_v+\alpha^2(\f'_u\V\f_v)^2}.
\end{equation*}

Denote
\begin{equation*}
d_v=\f'_v\V\f_v, \: d_u=\f'_u\V\f_u, \: d_{uv}=\f'_u\V\f_v.
\end{equation*}
\begin{equation*}
a_v:=\f'_v\V^2\f_v, \: a_u:=\f'_u\V^2\f_u, \: a_{uv}:=\f'_u\V^2\f_v,
\end{equation*}
and
\begin{equation}\label{eq:constants}
A=a_v-a_u, \: B=2d_{uv}a_{uv}-d_u a_v-d_v a_u, \: C=d_v-d_u, \: D=d_u d_v-d_{uv}^2.
\end{equation}
Then
\begin{equation*}
-\tr(\M^{-1}_\alpha)=-\tr(\V)+\frac{\alpha A+\alpha^2 B}{1+\alpha C-\alpha^2 D} =: h(\alpha),
\end{equation*}

The function $h$ is smooth and concave on $I$. Additionally, if $\f_u \neq \pm \f_v$, then it is also \emph{strictly} concave, which follows from the smoothness and strict concavity of $-\tr(\M^{-1})$ on the set of all positive definite matrices $\M$ (see \cite{Pazman}, Prop. IV.3).

 Note that the derivative of $h(\alpha)$ in $\alpha \in I$ is
\begin{equation}\label{eq:hder}
\frac{\mathrm{d} h(\alpha)}{\mathrm{d} \alpha}=
\frac{A+2\alpha B+\alpha^2(AD+BC)}{(1+\alpha C-\alpha^2 D)^2}.
\end{equation}
We will find $\alpha^* \in \bar{I}:=[-w_v,w_u]$ maximizing the continuous extension $\bar{h}:\bar{I} \to \R \cup \{-\infty\}$ of $h$.

\begin{lemma}\label{eq:lemadiscr}
 $B^2-A(AD+BC) \geq 0$.
\end{lemma}
\begin{proof}
It is straightforward to verify that $B^2-A(AD+BC)$ is equal to
\begin{equation}\label{eq:discrim}
 \left[(d_u+d_v)^2-4d_{uv}^2\right]\left[a_ua_v-a_{uv}^2\right]+\left[d_{uv}(a_u+a_v)-a_{uv}(d_u+d_v)\right]^2.
\end{equation}
Using the Cauchy-Schwarz inequality and the AM-GM inequality, we obtain $4d_{uv}^2 \leq 4d_ud_v \leq (d_u+d_v)^2$ and $a_{uv}^2 \leq a_ua_v$. Therefore, \eqref{eq:discrim} is non-negative.
\end{proof}

\begin{proposition}
 If $AD=-BC$ and $B \neq 0$ let $\alpha^*_s=-\frac{A}{2B}$, and if $AD \neq -BC$ let
\begin{equation*}
\alpha^*_s=-\frac{B+\sqrt{B^2-A(AD+BC)}}{AD+BC}.
\end{equation*}
If $\alpha^*_s \in I$ then $\alpha^*_s$ maximizes $\bar{h}$ on $\bar{I}$. If $AD=-BC$ and $B=0$ or if $\alpha^*_s \notin I$, let
 \begin{equation}\label{eq:AoptStepN}
    \alpha^*_n=\begin{cases}
    w_u & \mbox{ if } A>0,\\
    0 & \mbox{ if } A=0,\\
    -w_v & \mbox{ if } A<0.
    \end{cases}
\end{equation}
Then, $\alpha^*_n$ maximizes $\bar{h}$ on $\bar{I}$.
\end{proposition}
\begin{proof}
	The case $\f_u=\pm \f_v$ means that the numerator of \eqref{eq:hder} is zero, i.e., any $\alpha_s^* \in \bar{I}$ is a maximizer. Therefore, we will consider the non-trivial case $\f_u \neq \pm \f_v$, which implies that $h$ is strictly concave on $I$. First, note that \eqref{eq:hder} is equal to $A$ for $\alpha=0$.
	
	Let $AD=-BC$, $B\neq 0$ and $\alpha^*_s=-A/(2B) \in I$. Then, based on \eqref{eq:hder}, $\alpha^*_s$ is a stationary point of $h$ in $I$. Therefore, it must maximize $\bar{h}$ on $\bar{I}$.
	
	Let $AD\neq -BC$ and $\alpha^*_s \in I$. Taking Lemma \ref{eq:lemadiscr} into account, we see that the quadratic function $g(\alpha)$ in the numerator of $\eqref{eq:hder}$ has real roots on $\R$ given by
	\begin{equation*}
	\alpha_{1,2}=-\frac{B\pm\sqrt{B^2-A(AD+BC)}}{AD+BC},
	\end{equation*}
	where $\alpha^*_s=\alpha_1$. Since $h$ is strictly concave on $I$, it has at most one stationary point on $I$, which is the unique maximizer on $I$.
	
	If $AD=-BC$, $B=0$ or if $AD=-BC$, $B \neq 0$, $\alpha^*_s \notin I$, then \eqref{eq:hder} has the same sign as $A$ for any $\alpha \in I$, which means that \eqref{eq:AoptStepN} indeed provides a maximizer of $\bar{h}$ on $\bar{I}=[-w_v,w_u]$.
	
	Finally, let $AD \neq -BC$ and $\alpha^*_s \notin I$. It is enough to prove that for $A \neq 0$ none of $\alpha_1, \alpha_2$ belongs to $I$. If so, then the sign of the derivative \eqref{eq:hder} will be equal to the sign of $A$ on the entire interval $I$, which implies \eqref{eq:AoptStepN}. Since we assume that $\alpha_1=\alpha^*_s \notin I$, it is enough to prove that $\alpha_2 \notin I$.
	
	Let $A>0$. Consider mutually exclusive cases a) $AD<-BC$ and b) $AD>-BC$. If a) then  $\alpha_2 \leq 0$, it means that $\alpha_2$ cannot be a stationary point of $h$ on $I$ (since the derivative of $h$ in $0$ is positive and $h$ is concave on $I$), hence $\alpha_2 \notin I$. If b), then $\alpha_2>0$, but $\alpha_2 > \alpha_1 \geq 0$. As $\alpha_1 \notin I$, we have $\alpha_2 \notin I$. If $A < 0$, we can use an analogous argument.
\end{proof}

Thus, we can determine a maximizer $\alpha^*$ of $\bar{h}$ on $\bar{I}$ as follows:
\begin{enumerate}
	\item Set $G:=AD+BC$
	\item If $G=0$ and $B \neq 0$ set $r:=-A/(2B)$. If $-w_v < r < w_u$, return $\alpha^*=r$.
	\item If $G \neq 0$ set $r:=-(B+\sqrt{B^2-AG})/G$. If $-w_v < r < w_u$, return $\alpha^*=r$.
	\item If $A>0$ return $\alpha^*=w_u$, else if $A<0$ return $\alpha^*=-w_v$, else return $\alpha^*=0$.
\end{enumerate}

\end{document}